\documentclass[conference, letterpaper]{IEEEtran}
\usepackage{epsfig}
\usepackage{times}
\usepackage{float}
\usepackage{afterpage}
\usepackage{amsmath}
\usepackage{amstext}
\usepackage{amssymb,bm}
\usepackage{latexsym}
\usepackage{color}
\usepackage{graphicx}
\usepackage{amsmath}
\usepackage{amsthm}
\usepackage{graphicx}
\usepackage{subfigure}
\usepackage{booktabs}
\usepackage{multicol}
\usepackage{lipsum}
\usepackage{dblfloatfix}
\usepackage{mathrsfs}
\usepackage{cite}
\usepackage{tikz}
\usepackage{pgfplots}
\pgfplotsset{compat=newest}
 \usepackage{enumitem}
\usepackage{makecell}

\allowdisplaybreaks

\newcommand{\expect}[1]{{\mathbb{E}}\left[#1\right]}
\newcommand{\expcnd}[2]{{\mathbb{E}}\left[ #1 \;\middle|\; #2\right]}

\newcommand{\lr}[1]{\left( #1 \right)}

\newcommand{\bbC}{\mathbb{C}}
\newcommand{\rmd}{\mathrm{d}}
\newcommand{\bbE}{\mathbb{E}}\newcommand{\rme}{\mathrm{e}}

\newcommand{\bbN}{\mathbb{N}}\newcommand{\rmN}{\mathrm{N}}

\newcommand{\bbP}{\mathbb{P}}

\newcommand{\bbR}{\mathbb{R}}

\newcommand{\rmZ}{\mathrm{Z}}

\newcommand{\sfA}{\mathsf{A}}
\newcommand{\sfB}{\mathsf{B}}

\newcommand{\sfN}{\mathsf{N}}

\newcommand{\sfR}{\mathsf{R}}

\newcommand{\cA}{\mathcal{A}}

\newcommand{\cD}{\mathcal{D}}

\newcommand{\D}{D}
\newcommand{\kl}[2]{{\D}\left(\left.#1 \, \right\| #2 \right)}

\newcommand{\supp}{{\mathsf{supp}}}

\theoremstyle{mystyle}
\newtheorem{theorem}{Theorem}
\theoremstyle{mystyle}
\newtheorem{lemma}{Lemma}
\theoremstyle{mystyle}
\newtheorem{prop}{Proposition}
\theoremstyle{mystyle}
\theoremstyle{mystyle}
\theoremstyle{remark}
\theoremstyle{mystyle}
\theoremstyle{mystyle}
\theoremstyle{mystyle}
\theoremstyle{discussion}
\theoremstyle{mystyle}
\theoremstyle{mystyle}

 \usepackage{enumitem}

\begin{document}

\title{Improved Bounds on the Number of Support Points of the Capacity-Achieving Input for Amplitude Constrained Poisson Channels}

\author{
\IEEEauthorblockN{ Luca Barletta$^{*}$, Alex Dytso$^{\dagger}$, and Shlomo Shamai (Shitz)$^{**}$}
$^{*}$ Politecnico di Milano, Milano, 20133, Italy. Email: luca.barletta@polimi.it \\
$^{\dagger}$ Qualcomm Flarion Technologies, Bridgewater,  NJ 08807, USA.
Email: odytso2@gmail.com \\
$^{**}$  Technion -- Israel Institute of Technology, Haifa 32000, Israel. E-mail: sshlomo@ee.technion.ac.il}

\maketitle
\begin{abstract}
This work considers a discrete-time Poisson noise channel with an input amplitude constraint $\sfA$ and a dark current parameter $\lambda$. It is known that the capacity-achieving distribution for this channel is discrete with finitely many points. Recently, for $\lambda=0$, a lower bound of order $\sqrt{\sfA}$ and an upper bound of order $\sfA \log^2(\sfA)$ have been demonstrated on the cardinality of the support of the optimal input distribution.

In this work, we improve these results in several ways. First, we provide upper and lower bounds that hold for non-zero dark current. Second, we produce a sharper upper bound with a far simpler technique. In particular, for $\lambda=0$, we sharpen the upper bound from the order of $\sfA \log^2(\sfA)$ to the order of $\sfA$. Finally, some other additional information about the location of the support is provided. 
\end{abstract}
\section{Introduction}
The conditional
probability mass function (pmf) of the output random variable
$Y$ given the input $X$ that specifies the discrete-time Poisson channel is given by:
\begin{equation}
P_{Y|X}(y|x) =\frac{1}{y!} (x+\lambda)^y \rme^{-(x+\lambda) }, \, x\ge 0, \lambda \ge 0,  y \in \bbN_0.  \label{eq:Poisson_Channel}
\end{equation}
In the above, we used the convention that $0^0 =1$.  The parameter $\lambda  \ge 0$ is known as the \emph{dark current}. 

The goal of this work is to study the capacity of the channel in \eqref{eq:Poisson_Channel} where the input is subject to an amplitude constraint, that is 
\begin{equation}
    C(\sfA, \lambda) =\max_{ P_X :\: 0 \le X \le \sfA } I(X;Y). \label{eq:Cap_def}
\end{equation}
More specifically, we are interested in investigating the structure of the \emph{capacity-achieving distribution}, which will be denoted by $P_{X^\star}$. 

\subsection{Background}
\label{subsec:background}
The discrete-time Poisson noise channel is suited to model low intensity, direct detection optical communication channels \cite{gordon1962quantum} and molecular communications based on particle-intensity \cite{farsad2020capacities}.  In this work, we are concerned with the discrete-time channel; however, there
also exists a large literature on continuous-time channels, and the interested reader is referred to a survey
\cite{verdu1999poisson}.

The initial substantial investigation into the capacity-achieving distribution for the Poisson channel was conducted in \cite{mceliece1979practical}, where it was shown that when the dark current is absent ($\lambda=0$) the support of an optimal input distribution for any $\sfA$ can contain
at most one mass point in the interval $(0, 1)$. Moreover, for any $\sfA < 1$, it was shown that the optimal
input distribution consists of two mass points at $0$ and $\sfA$. In \cite{shamai1990capacity}, it was shown that the optimizing input distribution
is unique and discrete with finitely many mass points. Critically, the result in \cite{shamai1990capacity} was also shown for non-zero dark current. For other relevant investigations from properties of the capacity-achieving distribution the interested reader is referred to \cite{cao2014capacity,cao2014capacityPart2}. 

It is important to note that original techniques,  for showing discreetness of the input distribution, which were developed in the context of the additive Gaussian noise channel in \cite{smith1971information}, are typically non-constructive and cannot provide  additional information about the size or location of the support. 
Recently, in \cite{dytso2021properties}, for the case of $\lambda =0$,  the results of \cite{shamai1990capacity} were sharpened and firm bounds on the support size have been provided. Specifically, \cite{dytso2021properties} derived a lower bound of order $\sqrt{\sfA}$  and an upper bound of order $\sfA \log^2(\sfA)$. The proof of the upper bound in \cite{dytso2021properties} relied on the oscillation theorem due to Karlin \cite{karlin1957polya} and several additional rather involved inequalities. This work takes a different approach which has several advantages. First, the proof is simpler. Second, the proof generalizes to the case of non-zero dark current. Third, the  resulting bounds are tighter. 

Capacity upper bound have also received considerable attention and can be found in \cite{martinez2007spectral,martinez2008achievability,lapidoth2011discrete,wang2014refined,wang2014impact, mceliece1979practical,verdu1990asymptotic,martinez2007spectral,cheraghchi2018improved,cheraghchi2020non,lapidoth2009capacity}. 

\subsection{Organization}
\label{subsec:organization}
The remainder of the paper is organized as follows. 
Sec.~\ref{sec:preliminaries} defines notation, reviews the Karush-Kuhn-Tucker (KKT) conditions, derives some useful estimation theoretic identities, proposes bounds on a ratio of cumulants, and presents Tijdeman's lemma.  
Sec.~\ref{sec:main_result} presents our main results, whose proof is provided in Sec.~\ref{sec:proofs}. 
Finally, Sec.~\ref{sec:conclusion} concludes the paper.

\section{Preliminaries }
\label{sec:preliminaries}
\subsection{Notation}
All logarithms are to the base $\rme$. Deterministic scalar quantities are denoted by lower-case letters and random variables are denoted by uppercase letters.  
For a random variable $X$ and every measurable subset $\cA \subseteq \bbR$ the probability distribution is written as $P_{X}(\cA) = \bbP[X \in \cA]$. The support set of $P_X$ is
\begin{align}
\supp(P_{X})=\{x:&  \text{ $P_{X}( \mathcal{D})>0$ for every open set $ \mathcal{D} \ni x $ } \}. \nonumber
\end{align} 
When $X$ is discrete, we write $P_X(x)$ for $P_X(\{x\})$, i.e., $P_X$ is a probability mass function (pmf). The relative entropy of the distributions $P$ and $Q$ is $\kl{P}{Q}$. 

Given a function $f: \bbC \mapsto \bbC$ and a set $\cA \subseteq \bbC$, define the set of zeros of $f$ in $\cA$ as
\begin{equation}
\rmZ\left(\cA; f\right) =    \{z : f(z) =0 \} \cap \cA.  
\end{equation}
We denote the cardinality of $\rmZ\left(\cA; f\right)$ by $\sfN\left(\cA; f\right)$. 

\subsection{KKT Equations}
The key to studying properties of the optimal input distribution are the following KKT conditions which relate the support of $P_{X^\star}$ to finding zeros of a certain function defined through the relative entropy \cite{shamai1990capacity}. 
\begin{lemma}\label{lem:KKT}
    $P_{X^\star}$ maximizes \eqref{eq:Cap_def} if and only if
    \begin{align}
        i(x; P_{X^\star}) &\le C(\sfA,\lambda), \, x \in [0, \sfA] \\
        i(x; P_{X^\star}) &= C(\sfA,\lambda), \, x \in \supp(P_{X^\star}) \label{eq:KKT_equality_condition}
    \end{align}
    where
    \begin{equation}
    i(x; P_{X^\star}) = \kl{P_{Y|X}(\cdot|x)}{P_{Y^\star}}. 
    \end{equation}
\end{lemma}

\subsection{Estimation Theoretic Identities}
Next, we present several  derivative identities between  relative entropy and estimation theoretic quantities such as the conditional mean and conditional cumulant, which will be useful in our proofs. 

    The \emph{conditional cumulant generating function} is defined as
    \begin{align}
    K_{X|Y=y}(t)& =\log \bbE[ \rme^{tX}|Y =y]\\
    &=\log \bbE[ (X+\lambda)^y\rme^{-X(1-t)-\lambda}] + \log(y! P_Y(y) ).
    \end{align}

    The first conditional cumulant is
    \begin{align}
    \kappa_1(X|Y=y) &=  \frac{\rmd}{\rmd t} K_{X|Y=y}(t) 
 \Big |_{t=0}  \\
 &=\frac{\bbE[X(X+\lambda)^{y} \rme^{-X-\lambda}]}{\bbE[(X+\lambda)^{y} \rme^{-X-\lambda}]} \\
    &= \frac{\expect{(X+\lambda)^{y+1}\rme^{-(X+\lambda)}}}{\expect{(X+\lambda)^{y}\rme^{-(X+\lambda)}}}-\lambda.
    \end{align}

We next show some useful derivative identities.    
\begin{lemma}\label{lem:Derivatives_K}
    Let 
    \begin{equation}
    i(x; P_X) = G(x) +(x+\lambda) \log (x+\lambda) -(x+\lambda), \, x \ge 0
    \end{equation}
Then,
\begin{itemize}
\item The function $G$ and its first and second derivative are given by 
\begin{align}
    G(x) &=  \sum_{y=0}^\infty P_{Y|X}(y|x) \log  \frac{1}{y! P_{Y}(y)}, \,  x \ge 0 \\
    G'(x) &= \sum_{y=0}^\infty P_{Y|X}(y|x) \log  \frac{1}{\bbE[X+\lambda|Y =y]  }, \, x>0 
    \end{align}
    \begin{align}
    &G''(x) 
 = \sum_{y=0}^\infty P_{Y|X}(y|x) \log  \frac{\bbE^2[X +\lambda|Y =y]  }{\bbE[ (X+\lambda)^2|Y =y]  } \\
& = \sum_{y=0}^\infty P_{Y|X}(y|x) \log  \frac{\bbE[X+\lambda|Y =y]  }{\bbE[ X+\lambda|Y =y+1]  }  \\
&=\sum_{y=0}^\infty P_{Y|X}(y|x)  \log\frac{\kappa_1(X|Y=y)+\lambda}{\kappa_1(X|Y=y+1)+\lambda}, \, x>0.
\end{align}
\item The second derivative of $i(x;P_X)$ is given by 
\begin{align}
    i''(x; P_X) &= \sum_{y=0}^\infty P_{Y|X}(y|x) \log  \frac{x+\lambda  }{\bbE[ X+\lambda|Y =y+1]  } \notag\\
   & \quad -i'(x; P_X)  +\frac{1}{x+\lambda},\qquad x  >0.  \label{eq:second_deriavtive_i_den}
\end{align}
\end{itemize}
\end{lemma}
\begin{proof}
  See Appendix~\ref{app:lem:Derivatives_K}.
\end{proof}

We conclude this part presenting bounds on the ratio of conditional cumulants. 

\begin{lemma}\label{lem:bound_cumulant}
    For all $\sfA>0$ and $y \in \bbN_0$ we have
    \begin{equation}
        \frac{\kappa_1(X^\star|Y=y)+\lambda}{\kappa_1(X^\star|Y=y+1)+\lambda} \le 1, \qquad \lambda\ge 0,
    \end{equation}
    and 
    \begin{align}
         \frac{\kappa_1(X^\star|Y=y)+\lambda}{\kappa_1(X^\star|Y=y+1)+\lambda} \ge \left\{\begin{array}{cc}
            \frac{\lambda}{\sfA+\lambda} & \lambda>0, \\
         \frac{\rme^{-\sqrt{2(\log(\sfA)-1)}}}{\sfA}    & \lambda=0.
        \end{array} \right.
    \end{align}
\end{lemma}
\begin{proof}
    See Appendix~\ref{app:lem:bound_cumulant}.
\end{proof}

\subsection{Counting Zeros of Complex Analytic Functions}

The key complex analytic tool that we will be using to count the number of zeros is the following lemma. 
\begin{lemma}[Tijdeman's Number of Zeros Lemma \cite{Tijdeman1971number}]\label{lem:number of zeros of analytic function}
 Let $\sfR, v, t$ be positive numbers such that $v>1$. For the complex valued function $f\neq  0$ which is analytic on $|z|\le(vt+v+t)\sfR$, its number of zeros $  \rmN(\cD_{\sfR};f)$ within the disk $\cD_{\sfR} = \{z\colon |z|\le \sfR\} $ satisfies
\begin{equation}
 \rmN(\cD_{\sfR};f) \le \frac{1}{\log v} \log \left( \frac{\max_{|z|\le (vt+v+t)\sfR } |f(z)|  }{ \max_{|z|\le t\sfR} |f(z)|} \right) \text{.} \label{eq:Tijdeman}
\end{equation}
\end{lemma}

\section{Main Result and Discussion}
\label{sec:main_result}

The main  results of this paper are summarized in the following theorem. 
\begin{theorem}
    Let $P_{X^\star}$ be the capacity-achieving input of the amplitude constrained Poisson channel. Then,
    \begin{enumerate}[label=(\roman*)]
        \item \emph{(On the Location of Interior Support Points)} For $\sfA+\lambda> \rme$ and $x^\star \in \supp(P_{X^\star})\setminus\{0,\sfA\}$ we have
        \begin{align}
 \rme^{ - \sqrt{2(\log(\sfA+\lambda)-1)} } 
 &\le  (\sfA +\lambda) \rme^{ W_{-1}(-\frac{1}{\sfA+\lambda} )}  \\
 &\le   x^\star+\lambda \\
 &\le  (\sfA +\lambda) \rme^{ W_0(-\frac{1}{\sfA+\lambda} )} \le \sfA +\lambda -1. 
\end{align}

\item \emph{(Upper Bound on the Number of Support Points)}
\begin{itemize}[leftmargin =-0.2cm]
    \item For $\sfA+\lambda < \rme$ we have
    \begin{equation}
        |\supp(P_{X^\star})|=2.
    \end{equation}
    \item For $\lambda=0$ and $\sfA > \rme$ we have
    \begin{align}
    &|\supp{(P_{X^\star})}| \leq\Big\lfloor3  \nonumber\\
    &\left.+ \log\left(1+ \sfA \rme^{2\rme\sfA+1}\left(\log(\sfA)+\sqrt{2(\log(\sfA)-1)}\right) \right)  \right\rfloor \label{eq:upper_bound_num_points}
\end{align}
\item For $\lambda>0$ and $\sfA+\lambda > \rme$ we have
\begin{align}
    &|\supp{(P_{X^\star})}| \leq \Big\lfloor 3 \nonumber\\
    & \hspace{-0.1cm}+ \log\left( 
 \hspace{-0.02cm}1+ \left(\rme(\sfA+\lambda)+2\lambda\hspace{-0.02cm}\right) \rme^{2(\rme(\sfA+\lambda)+\lambda)}\log\frac{\sfA+\lambda}{\lambda} \right) \Big\rfloor
\end{align}
\end{itemize}

    \item \emph{(Lower Bound on the Number of Support Points)}
    For all $\sfA > 0$ and $\lambda\ge 0$, we have
    \begin{equation}
        |\supp(P_{X^\star})| \ge \left\lceil\max\{2, \rme^{C(\sfA, \lambda)}\}\right\rceil.
    \end{equation}
    \end{enumerate}
\end{theorem}
We can get an explicit lower bound on the number of support points by using the capacity bound  in~\cite[Th.~4]{lapidoth2009capacity}:
\begin{align}
    \rme^{C(\sfA, \lambda)} \ge \sqrt{\frac{2\sfA}{\pi \rme^3}} \left(1+\frac{3}{\sfA}\right)^{1+\frac{\sfA}{3}} \rme^{-\sqrt{\frac{\lambda+\frac{1}{12}}{\sfA}}\left(\frac{\pi}{4}+\frac{1}{2}\log 2\right)}
\end{align}
which grows as $\sqrt{\sfA}$ for $\sfA\to\infty$. The upper bound on the support size in \eqref{eq:upper_bound_num_points} grows as $\sfA$ for $\sfA\to\infty$ and for any finite $\lambda$.

\begin{figure}
    \centering
    \includegraphics[width=\columnwidth]{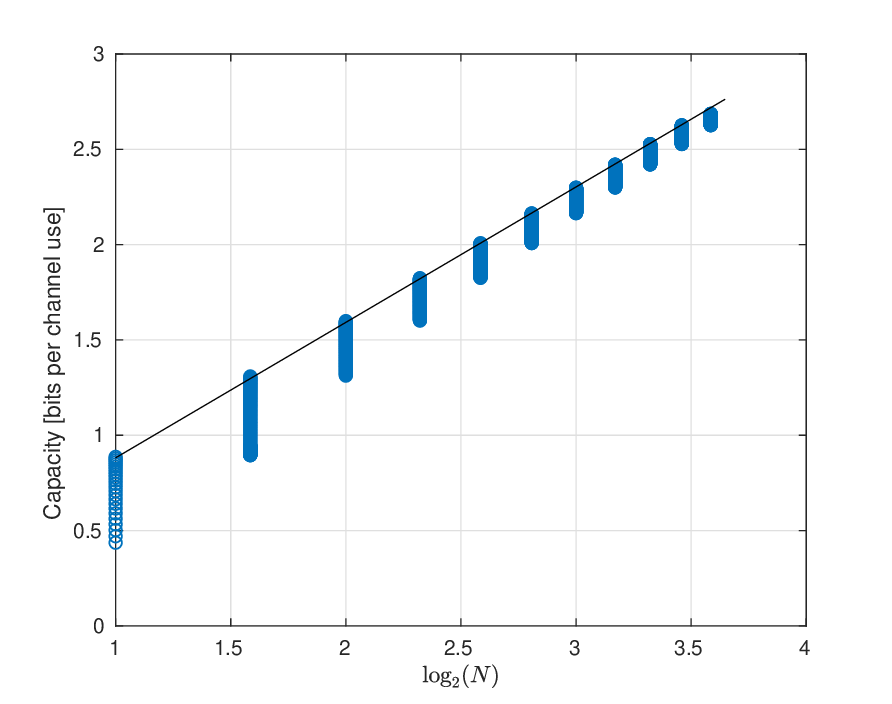}
    \caption{Capacity $C(\sfA,0)$ vs size of the optimal input distribution $N=|\supp(P_{X^\star})|$ for the case $\lambda=0$. Each circle denotes one value of $\sfA$. Solid line: interpolation of the values of $\sfA$ for which the number of support points increases.}
    \label{fig:capacity_vs_support_size}
\end{figure}

In Fig. \ref{fig:capacity_vs_support_size}, for $\lambda=0$ and several values of $\sfA$, we present $C(\sfA,0)$ versus the optimal number of support points $N$. The values of $C$ and $N$ are determined as described in \cite{barletta2022Poisson}.

The solid line in Fig. \ref{fig:capacity_vs_support_size} interpolates the points before a change in the value of $N$. This interpolating curve serves as an estimate of the relationship between $N$ and $C(\sfA,0)$. Specifically, we derive the asymptotic relationship:
\begin{equation}
C(\sfA,0) \approx 0.7 \log_2 N \text{ [bpcu]}, \quad N \to \infty,
\end{equation}
where `bpcu' stands for bits per channel use. Since the capacity scaling is $C(\sfA,0) \approx 0.5\log_2 \sfA$ as $\sfA\to \infty$ (as shown in \cite{lapidoth2009capacity}), we anticipate an asymptotic relationship of $N \approx \sfA^{\frac{0.5}{0.7}}$. This approximation suggests that both our proposed upper and lower bounds on $N$ require further refinement.

\section{Proofs}
\label{sec:proofs}
This section is dedicated to the proofs. 
\subsection{ On Location of Interior Support Points }
To produce bounds on the support size we also need additional information about the location of interior support points, which is provided next. 
\begin{prop}\label{prop:locations}
    For $\sfA+\lambda <\rme$, $|\supp(P_{X^\star})| =2$. For  $|\supp(P_{X^\star})|  \ge 3$,  we have that $\sfA+\lambda >\rme$ and
    \begin{align}
 \rme^{ - \sqrt{2(\log(\sfA+\lambda)-1)} } 
 &\le  (\sfA +\lambda) \rme^{ W_{-1}(-\frac{1}{\sfA+\lambda} )}  \\
 &\le   x^\star+\lambda \\
 &\le  (\sfA +\lambda) \rme^{ W_0(-\frac{1}{\sfA+\lambda} )} \le \sfA +\lambda -1. 
\end{align}
\end{prop}
\begin{proof} 
See Appendix~\ref{app:prop:locations}.
\end{proof}

\subsection{Upper Bound on the Number of Support Points}
The KKT conditions of Lemma~\ref{lem:KKT} imply that
\begin{equation}
    \supp(P_{X^\star}) \subseteq \rmZ\left([0, \sfA]; \kl{P_{Y|X}(\cdot|x)}{P_{Y^\star}}-C \right).
\end{equation}
Switching to the cardinalities of the sets, we have
\begin{align}
    |\supp{(P_{X^\star})}| &\leq 
    \mathrm{N}\lr{ \, [0,\sfA]; \, \kl{P_{Y|X}(\cdot|x)}{P_{Y^\star}}-C \, }  \\
    &=\mathrm{N}\lr{ \, (0,\sfA]; \, \kl{P_{Y|X}(\cdot|x)}{P_{Y^\star}}-C \, }+1 \label{eq:0_always_mass_point} \\
    &\le \mathrm{N}\lr{ \, (0,\sfA]; \, G'(x)+\log(x+\lambda) \, }+2 \label{eq:use_Rolle} \\
    &\le \mathrm{N}\lr{ \, (0,\sfA]; \, G''(x)+\frac{1}{x+\lambda} \, }+3 \\
    &= \mathrm{N}\lr{ \, (0,\sfA]; \, (x+\lambda)G''(x)+1 \, }+3 \label{eq:multiply_same_zeros} \\
    &= \mathrm{N}\lr{ \, (0,\sfA]; \, g(x) \, }+3 \label{eq:introduce_g}
    \end{align}
where~\eqref{eq:0_always_mass_point} holds because $0 \in \supp(P_{X^\star})$ (see~\cite{cao2014capacity}); \eqref{eq:use_Rolle} follows from Rolle's theorem (see~\cite[Lemma~3]{dytso2019capacity}) and from Lemma~\ref{lem:Derivatives_K}; and \eqref{eq:multiply_same_zeros} holds because multiplying the function by $x+\lambda$ does not change the number of zeros in $(0,\sfA]$. In~\eqref{eq:introduce_g}, we have introduced the function
\begin{align}
    g(x) &= (x+\lambda)G''(x)+1 \\
    &= (x+\lambda)\expcnd{\log\frac{\kappa_1(X^\star|Y)+\lambda}{\kappa_1(X^\star|Y+1)+\lambda}}{X=x}+1.
\end{align}
Let us now consider the complex analytic extension of the real function $g$:
\begin{align}
    \Breve{g}(z) = 1+(z+\lambda) \sum_{y=0}^\infty &\frac{(z+\lambda)^{y}}{y!}\rme^{-(z+\lambda)}\nonumber \\
    &\qquad
 \cdot\log \frac{\kappa_1(X^\star|Y=y)+\lambda}{\kappa_1(X^\star|Y=y+1)+\lambda}
\end{align}
which is analytic for all $z \in \bbC$.


Next, we compute bounds on the maximum value of $|\Breve{g}(z)|$ in a closed disk of radius $\sfB$. For the upper bound, we have
\begin{align}
    &\max_{|z|\le \sfB} |\Breve{g}(z)| = \max_{|z|\le \sfB} \left| 1+ (z+\lambda) \sum_{y=0}^\infty \frac{(z+\lambda)^y}{y!}\rme^{-(z+\lambda)}\right. \nonumber\\
    &\qquad\qquad\qquad
 \left.\cdot\log \frac{\kappa_1(X^\star|Y=y)+\lambda}{\kappa_1(X^\star|Y=y+1)+\lambda}\right| \\
 &\le 1+ \sum_{y=0}^\infty \frac{(\sfB+\lambda)^{y+1}}{y!}\rme^{\sfB-\lambda}
  \log \frac{\kappa_1(X^\star|Y=y+1)+\lambda}{\kappa_1(X^\star|Y=y)+\lambda}\label{eq:apply_triangle_ineq} \\
 &\le 1+(\sfB+\lambda) \rme^{2\sfB}\cdot\left\{\begin{array}{cc}
     \log\frac{\sfA+\lambda}{\lambda}
     & \lambda>0, \\
     \log(\sfA)+\sqrt{2(\log(\sfA)-1)} & \lambda=0
 \end{array} \right. \label{eq:apply_bound_ratio_cumulants}
\end{align}
where~\eqref{eq:apply_triangle_ineq} follows from the triangle inequality, from $|z|\le \sfB$, and because by Lemma~\ref{lem:bound_cumulant} the ratio of cumulants is less than $1$; and \eqref{eq:apply_bound_ratio_cumulants} follows from the lower bound on the ratio of cumulants of Lemma~\ref{lem:bound_cumulant}.

For the lower bound on the maximum value of $|\Breve{g}(z)|$, for $\sfB\ge \lambda$ we have
\begin{equation}
    \max_{|z|\le \sfB} |\Breve{g}(z)| \ge |\Breve{g}(-\lambda)|= 1.
\end{equation}

Let us now apply Tijdeman's lemma for the case $\lambda =0$. By choosing $v=\rme$ and $t=0$, we get
\begin{align}
 &\mathrm{N}\lr{ \, (0,\sfA]; \, g(x) \, } \le
    \mathrm{N}\lr{ \, {\cal D}_{\sfA}; \, \breve{g}(z) \, } \\
    &\le \log\left(1+ \sfA \rme^{2\rme\sfA+1}\left(\log(\sfA)+\sqrt{2(\log(\sfA)-1)}\right) \right).
\end{align}

For the case $\lambda>0$, by applying Tijdeman's lemma with $v=\rme$ and  $t=\frac{\lambda}{\sfA}$ we get:
\begin{align}
 &\mathrm{N}\lr{ \, (0,\sfA]; \, g(x) \, } \le
    \mathrm{N}\lr{ \, {\cal D}_{\sfA}; \, \breve{g}(z) \, } \\
    &\le \log\left(1+ \left(\rme(\sfA+\lambda)+2\lambda\right) \rme^{2(\rme(\sfA+\lambda)+\lambda)}\log\frac{\sfA+\lambda}{\lambda} \right).
\end{align}

Putting everything together, we get
\begin{align}
    &|\supp{(P_{X^\star})}| \leq 3 \nonumber\\
    &+\left\{
    \begin{array}{lc}
       \log\left(1+ \left(\rme(\sfA+\lambda)+2\lambda\right) \rme^{2(\rme(\sfA+\lambda)+\lambda)}\log\frac{\sfA+\lambda}{\lambda} \right)  & \lambda>0 \\
       \log\left(1+ \sfA \rme^{2\rme\sfA+1}\left(\log(\sfA)+\sqrt{2(\log(\sfA)-1)}\right) \right)  & \lambda=0
    \end{array}
    \right.
\end{align}
which grows as $\sfA$ for $\sfA\to \infty$.

 \subsection{Lower Bound on the Number of Support Points}

There are several ways of finding lower bounds. We here offer an approach that also relies on the KKT and provides an exact but implicit expression for $N$.
We begin by recalling that for $P_X \to P_{Y|X} \to P_Y$ and $Q_X  \to P_{Y|X} \to Q_{Y}$, we have that
\begin{equation}
    \kl{P_X}{Q_X} = \kl {P_Y}{Q_Y} + D (P_{X|Y} \| Q_{X|Y} |P_Y)  \label{eq:KL_identity}
\end{equation}
where the conditional relative entropy is defined as
\begin{equation}
    D (P_{X|Y} \| Q_{X|Y} |P_Y) = \sum_{y=0}^\infty P_Y(y) \kl{P_{X|Y}(\cdot|y)}{Q_{X|Y}(\cdot|y)}
\end{equation}
\begin{lemma} Let $N:=| \supp(P_{X^\star})|$. Then,
    \begin{equation}
N = \frac{\rme^{  C(\sfA,\lambda) }}{  \frac{1}{N}  \sum_{ x^{\star} \in \supp(P_{X^{\star}}) }  \rme^{  -  D  \left(\delta_{x^\star} \| P_{X^\star|Y} |P_{Y|X }(\cdot|x^\star) \right) }}. 
\end{equation}
\end{lemma}
\begin{proof}
    Using the  KKT condition in \eqref{eq:KKT_equality_condition}, we have that for $x^{\star} \in \supp(P_{X^{\star}})$
\begin{align}
    C(\sfA,\lambda) &= \kl{P_{Y|X}(\cdot|x^\star)}{P_{Y^\star}}\\
    &=  \kl{ P_{Y_{x^\star} }  }{P_{Y^\star}} \label{eq:Def_output_P_Y}\\
    &= \kl{\delta_{x^\star} }{  P_{X^\star} } -    D (\delta_{x^\star} \| P_{X^\star|Y} |P_{Y_{x^\star} }) \label{eq:apply_KL_identity}\\
    &= \log \frac{1}{P_{X^\star} (x^\star)} -    D (\delta_{x^\star} \| P_{X^\star|Y} |P_{Y_{x^\star} }), \label{eq:last_step_P(x)}
\end{align}
where  \eqref{eq:Def_output_P_Y} follows by defining $\delta_{x^\star} \to P_{Y|X} \to P_{Y^\star}$;  and \eqref{eq:apply_KL_identity} follows by using \eqref{eq:KL_identity}.

By rearranging~\eqref{eq:last_step_P(x)} we arrive at  
\begin{equation}
    P_{X^\star}(x^\star) = \rme^{ - C(\sfA,\lambda) -  D  \left(\delta_{x^\star} \| P_{X^\star|Y} |P_{Y|X }(\cdot|x^\star) \right) }
\end{equation}
for $x^{\star} \in \supp(P_{X^{\star}})$.
Now summing over $x^{\star} \in \supp(P_{X^{\star}})$ we arrive at
\begin{equation}
1 =   \rme^{ - C(\sfA,\lambda) } \sum_{ x^{\star} \in \supp(P_{X^{\star}}) }  \rme^{  -  D  \left(\delta_{x^\star} \| P_{X^\star|Y} |P_{Y|X }(\cdot|x^\star) \right) }. \label{eq:Intermideiate_step}
\end{equation}
Rearranging \eqref{eq:Intermideiate_step}, we arrive at the following exact characterization of $N$:
\begin{equation}
N = \frac{\rme^{  C(\sfA,\lambda) }}{  \frac{1}{N}  \sum_{ x^{\star} \in \supp(P_{X^{\star}}) }  \rme^{  -  D  \left(\delta_{x^\star} \| P_{X^\star|Y} |P_{Y|X }(\cdot|x^\star) \right) }}. 
\end{equation}
\end{proof}

The proof of the lower bound now follows by using the non-negativity of relative entropy, which leads to 
\begin{equation}
    N \ge  \rme^{ C(\sfA,\lambda) }. 
\end{equation}

\section{Conclusion}
\label{sec:conclusion}
This work has focused on providing upper and lower bounds on the support size of the capacity-achieving distribution for the discrete-time Poisson channel. The paper both extends the bounds to the more general case with non-zero dark current and improves on previous bounds. Interestingly, our results suggest that the oscillation theorem, which results in tight bounds on the support size in the additive Gaussian noise channel \cite{dytso2019capacity}, does not appear to be the right tool for the Poisson channel. Instead, a direct application of complex analytic tools results in tighter bounds.

An interesting future direction is to find the exact order of the size of the support. The numerical results suggest that neither the upper nor lower bounds are tight. It would also be interesting to apply some of the current techniques to problems such as randomized identification  \cite{labidi2023information} which also involves a Poisson noise channel.

\begin{appendices}

    \section{Bound on a Moment Ratio}
  \begin{lemma}\label{lem:moment_ratio} Suppose that $f>0$ and $g$ is arbitrary. Then,
\begin{align}
        \sup_{X \in [a,b] }  \left | \frac{\bbE[g(X)]}{\bbE[  f(X) ] } \right|  = \max_{x \in [a,b]} \left| \frac{g(x)}{ f(x)} \right|.
    \end{align}
 \end{lemma}
 \begin{proof}
     Let 
    \begin{align}
        M :=  \max_{x \in [a,b]} \left| \frac{g(x)}{ f(x)} \right|. 
    \end{align}
    Then,
    \begin{align}
    | \bbE[g(X)] | &\le   \bbE[ |g(X)| ] \\
    & = \bbE\left[ \left| \frac{g(X)}{ f(X)} \right|  |f(X)| \right] \\
    &\le  M  \bbE[ | f(X)| ],
    \end{align}
    therefore,
    \begin{align}
   \sup_{X \in [a,b] }   \frac{ \left |\bbE[g(X)] \right|}{\bbE[ | f(X) |] }   =      \sup_{X \in [a,b] }  \left | \frac{\bbE[g(X)]}{\bbE[  f(X) ] } \right|
        \le  \max_{x \in [a,b]} \left| \frac{g(x)}{ f(x)} \right| \label{eq:moment_ratio_bound}
    \end{align}
    where the first equality holds due to the assumption $f>0$. 

    Note that equality in~\eqref{eq:moment_ratio_bound} is achievable by choosing $X$ to be a point mass concentrated at $x$ that attains the maximum. 
    \end{proof}

\section{Proof of Lemma~\ref{lem:bound_cumulant}} \label{app:lem:bound_cumulant}
For the upper bound, we have
 \begin{align}
 &\frac{\kappa_1(X^\star|Y=y)+\lambda}{\kappa_1(X^\star|Y=y+1)+\lambda} \notag\\
 &= \frac{\expect{(X^\star+\lambda)^{y+1}\rme^{-(X^\star+\lambda)}}^2}{\expect{(X^\star+\lambda)^{y}\rme^{-(X^\star+\lambda)}}\expect{(X^\star+\lambda)^{y+2}\rme^{-(X^\star+\lambda)}}}  \notag \\
 &\le 1 \label{eq:CS_ine_app}
    \end{align}
where \eqref{eq:CS_ine_app} is due to Cauchy-Schwarz inequality applied to the random variables $(X^\star+\lambda)^{\frac{y}{2}}\rme^{-\frac{X^\star+\lambda}{2}}$ and $(X^\star+\lambda)^{\frac{y+2}{2}}\rme^{-\frac{X^\star+\lambda}{2}}$. 

For the lower bound, we first consider the case $\lambda>0$. We have:
\begin{equation*}
\frac{\kappa_1(X^\star|Y=y)+\lambda}{\kappa_1(X^\star|Y=y+1)+\lambda} \hspace{-0.1cm}=  \frac{\expcnd{X^\star}{Y=y}+\lambda}{\expcnd{X^\star}{Y=y+1}+\lambda} \hspace{-0.1cm} \ge \frac{\lambda}{\sfA+\lambda}.
\end{equation*}
If $\lambda=0$, then
\begin{align*}
  \frac{1}{\kappa_1(X^\star|Y=y)} &=  \frac{\expect{(X^\star)^{y}\rme^{-X^\star}}}{\expect{(X^\star)^{y+1}\rme^{-X^\star}}}  \\
  &=\frac{\expcnd{(X^\star)^{y}\rme^{-X^\star}}{X^\star>0}}{\expcnd{(X^\star)^{y+1}\rme^{-X^\star}}{X^\star>0}} \\
  &\le \sup_{P_X:\: X\in[x_{\min},\sfA]} \frac{\expcnd{X^{y}\rme^{-X}}{X\ge x_{\min}}}{\expcnd{X^{y+1}\rme^{-X}}{X\ge x_{\min}}} \\
  &= \max_{x \in [x_{\min},\sfA]} \frac{1}{x} = \frac{1}{x_{\min}}
\end{align*}
where in last step we have used the result of Lemma~\ref{lem:moment_ratio}  and where 
\begin{equation}
    x_{\min} := \min \{\supp(P_{X^\star})\setminus \{0\} \} \ge \rme^{-\sqrt{2(\log(\sfA)-1)}}
\end{equation}
is a result of Proposition~\ref{prop:locations}. Also, we have
\begin{align}
    \kappa_1(X^\star|Y=y+1) &= \frac{\expect{(X^\star)^{y+2}\rme^{-X^\star}}}{\expect{(X^\star)^{y+1}\rme^{-X^\star}}} \\
    &\le \sup_{P_X:\: X\in[0,\sfA]} \frac{\expect{X^{y+2}\rme^{-X}}}{\expect{X^{y+1}\rme^{-X}}} \\
  &= \max_{x \in [0,\sfA]} x = \sfA
\end{align}
where we applied again Lemma~\ref{lem:moment_ratio}. Collecting the partial results, we get that for $\lambda =0$
\begin{equation} \frac{\kappa_1(X^\star|Y=y)}{\kappa_1(X^\star|Y=y+1)} \ge \frac{\rme^{-\sqrt{2(\log(\sfA)-1)}}}{\sfA}.
\end{equation}

 \section{Proof of Proposition~\ref{prop:locations} }
\label{app:prop:locations}

    From~\cite{cao2014capacity}, we know that $\{0,\sfA\}\in\supp(P_{X^\star})$ for all $\sfA>0$ and $\lambda \ge 0$.
Suppose that there exists $x^\star \in \supp(P_{X^\star}) \setminus \{0,\sfA\}$. Then  by the KKT condition in \eqref{lem:KKT}, $x^\star$ is a maximum of $x \mapsto i(x;P_{X^\star})$, hence we have that
\begin{equation}
i'(x^\star; P_{X^\star}) =0, \text{ and }
i''(x^\star; P_{X^\star}) \le 0. 
\end{equation} 
Then, using the above we have that
\begin{align}
 0&\ge i''(x^\star; P_{X^\star}) \\
&= \sum_{y=0}^\infty P_{Y|X}(y|x^\star) \log  \frac{x^\star+\lambda  }{\bbE[ X^\star+\lambda|Y =y+1]  } \nonumber \\
&\quad-i'(x^\star; P_{X^\star})  +\frac{1}{x^\star+\lambda} \label{eq:Using_sec_der_Expre}\\
&= \sum_{y=0}^\infty P_{Y|X}(y|x^\star) \log  \frac{x^\star+\lambda  }{\bbE[ X^\star+\lambda|Y =y+1]  } +\frac{1}{x^\star+\lambda}\\
& \ge    \log  \left( \frac{x^\star+\lambda  }{\sfA +\lambda }  \right) +\frac{1}{x^\star+\lambda} \label{eq:Using_bound_on_ce}
\end{align}
where \eqref{eq:Using_sec_der_Expre} follows from \eqref{eq:second_deriavtive_i_den}; and \eqref{eq:Using_bound_on_ce} follows from the bound $\bbE[X^\star |Y ] \le \sfA $.

We know seek to solve the inequality:
\begin{equation} \label{eq:inequality_to_solve}
0 \ge f(x) 
\end{equation}
where $f(x) =  \log  \left( \frac{x^\star+\lambda  }{\sfA +\lambda }  \right) +\frac{1}{x^\star+\lambda}$. 

Note that function $x\mapsto f(x)$  is decreasing for $x < 1-\lambda$ and increasing for $x > 1-\lambda$. Moreover, $f(x)$ has two zeros provided that $\sfA+\lambda >\rme$ and no zeros if $\sfA+\lambda <\rme$. 

Consequently, if $\sfA+\lambda <\rme$, $f(x)$ has no zeros and is a positive function, which leads to a contradiction. Therefore, for $\sfA+\lambda <\rme$, 
 $|\supp(P_{X^\star})| = 2$. This proves the first part of the result. 

For the second part, we assume that  $|\supp(P_{X^\star})| \ge 3$, which implies that $\sfA+\lambda >\rme$ and the function $x\mapsto f(x)$ has two zeros.  The exact solution to \eqref{eq:inequality_to_solve} is given in terms of Lambert-W functions:
\begin{equation}
  (\sfA +\lambda) \rme^{ W_{-1}(-\frac{1}{\sfA+\lambda} )}  \le   x^\star+\lambda \le  (\sfA +\lambda) \rme^{ W_0(-\frac{1}{\sfA+\lambda} )}.
\end{equation}

We now do further bounding to have  bounds that involve simpler functions. 

We start with a lower bound on $x^\star +\lambda$. By performing the substitution $x^\star+\lambda  =\rme^{t}$ we arrive at
\begin{equation}
f(t) = t -\log(\sfA+\lambda) + \rme^{-t}.
\end{equation}

Since $t \mapsto f(t)$ is decreasing for $t<0$, and noting that $\rme^{-t} \ge 1-t +\frac{t^2}{2}, t<0$, we arrive at
\begin{equation}
    f(t) \ge  -\log(\sfA +\lambda) +1 + \frac{t^2}{2}.
\end{equation}
This implies that the smallest zero of $f(t)$, denoted by  $t_0$,  is lower-bounded by 
\begin{equation}
    t_0  \ge  -\sqrt{2(\log(\sfA+\lambda)-1)}
\end{equation}
or equivalently that 
\begin{equation}
    x^\star  + \lambda \ge \rme^{ - \sqrt{2(\log(\sfA+\lambda)-1)} }. 
\end{equation}

We now provide an upper bound on $x^*+\lambda$. By using lower bound $\log(x) \ge 1-\frac{1}{x}$, we arrive at
\begin{equation}
0 \ge f(x)  \ge  1-\frac{\sfA +\lambda}{x^\star+\lambda} +\frac{1}{x^\star+\lambda}
\end{equation}
which implies that 
\begin{equation}
x^\star +\lambda \le \sfA+\lambda-1
\end{equation} 
This concludes the proof. 



   


\section{Proof of Lemma~\ref{lem:Derivatives_K}} \label{app:lem:Derivatives_K}

    We will need the following three identities \cite{PoissonCEprop}: for $\lambda \ge 0$
    \begin{align}
    \bbE[X|Y =y] &= \frac{ (y+1) P_Y(y+1)}{P_Y(y)} -\lambda, \, y \in \bbN_0  \label{eq:TuringIdentity} \\
   \frac{\rmd}{\rmd x}  \bbE[f(Y)|X=x] &= \bbE[ f(Y+1) -f(Y) | X=x], \, x >0
    \end{align}
    and
    \begin{align}
    &\expcnd{(X+\lambda)^k}{Y=y} \notag\\
    & \quad =  \prod_{i=0}^{k-1}  \expcnd{X+\lambda}{Y=y +k} ,  \, y \in \bbN_0 , k  \in \bbN_0 .\label{eq:ProductCEidentity}
    \end{align}
   
 Next, note that 
    \begin{equation}
    \kl{P_{Y|X}(\cdot| x)}{ P_Y} = G(x) + (x+\lambda) \log (x+\lambda) -(x+\lambda)  \label{eq:KL_as_G}
    \end{equation}
    where 
    \begin{equation}
        G(x) = \sum_{y=0}^\infty P_{Y|X}(y|x) \log  \frac{1}{y! P_{Y}(y)} , \, x \ge 0. 
    \end{equation}

Next, note that
\begin{align}
G'(x) &= \sum_{y=0}^\infty P_{Y|X}(y|x) \log  \frac{ P_{Y}(y)}{ (y+1) P_{Y}(y+1)} \\
&=\sum_{y=0}^\infty P_{Y|X}(y|x) \log  \frac{1}{\bbE[X|Y =y] +\lambda } \label{eq:G'_char}
\end{align}
where in the last step we have used the identity in \eqref{eq:TuringIdentity}.

Second, note that 
\begin{align}
G''(x) &= \sum_{y=0}^\infty P_{Y|X}(y|x) \log  \frac{\bbE[X 
 +\lambda|Y =y]  }{\bbE[X+\lambda|Y =y+1]  }\\
& = \sum_{y=0}^\infty P_{Y|X}(y|x) \log  \frac{\bbE^2[X +\lambda|Y =y]  }{\bbE[ (X+\lambda)^2|Y =y]  } \label{eq:secon_step_G''} \\
& = \sum_{y=0}^\infty P_{Y|X}(y|x) \log  \frac{1  }{\bbE[ X+\lambda|Y =y+1]  } - G'(x), \label{eq:third_step_G''}
\end{align}
where \eqref{eq:secon_step_G''} follows from identity in \eqref{eq:ProductCEidentity}; and \eqref{eq:third_step_G''} follows from \eqref{eq:G'_char}. 

The characterization of $i''(x;P_X)$ follow by combining  \eqref{eq:KL_as_G} and \eqref{eq:G'_char}. 
    \end{appendices}

\bibliographystyle{IEEEtran}
\bibliography{refs.bib}

\newpage 
\clearpage

\end{document}